\documentclass[11pt]{article}

\usepackage{amsfonts}
\usepackage{amsmath}
\usepackage{amssymb}
\usepackage{amsthm}

\usepackage{graphicx}

\theoremstyle{plain}
\newtheorem{theorem}{Theorem}

\newtheorem{lemma}[theorem]{Lemma}
\newtheorem{corollary}[theorem]{Corollary}

\theoremstyle{definition}
\newtheorem{example}[theorem]{Example}
\newtheorem{remark}[theorem]{Remark}

\numberwithin{equation}{section}
\numberwithin{theorem}{section}

\DeclareMathOperator{\diag}{diag}
\newcommand{\D}{\,\mathrm{d}}

\newcommand{\bs}[1]{{\boldsymbol{#1}}}

\newcommand{\R}{\mathbf{R}}
\newcommand{\N}{\mathbf{N}}

\newcommand{\Ham}[1]{H_{i}}

\begin{document}
%
%
%
%
%
%

\title{{Exact solutions for the selection--mutation equilibrium in the Crow--Kimura evolutionary model}}

\author{Yuri S. Semenov$^{1}$, Artem S. Novozhilov$^{{2},}$\footnote{Corresponding author: artem.novozhilov@ndsu.edu} \\[3mm]
\textit{\normalsize $^\textrm{\emph{1}}$Applied Mathematics--1, Moscow State University of Railway Engineering,}\\[-1mm]\textit{\normalsize Moscow 127994, Russia}\\[2mm]
\textit{\normalsize $^\textrm{\emph{2}}$Department of Mathematics, North Dakota State University, Fargo, ND 58108, USA}}

\date{}

\maketitle

\begin{abstract}
We reformulate the eigenvalue problem for the selection--mutation equilibrium distribution in the case of a haploid asexually reproduced population in the form of an equation for an unknown probability generating function of this distribution. The special form of this equation in the infinite sequence limit allows us to obtain analytically the steady state distributions for a number of particular cases of the fitness landscape. The general approach is illustrated by examples; theoretical findings are compared with numerical calculations.

\paragraph{\small Keywords:} Selection--mutation equilibrium, quasispecies model, Crow--Kimura model, error threshold, single peaked landscape
\paragraph{\small AMS Subject Classification:} Primary:  92D15; 92D25; Secondary: 15A18
\end{abstract}

\section{Introduction}
Selection and mutation are two main evolutionary forces that shape (together with recombination and genetic drift) the life histories of evolving populations. The mathematical theory of selection--mutation models is deep and elaborate (e.g., \cite{burger2000mathematical}), covering various biological assumptions, such as different mutation schemes, consequence of ploidy, mating systems, heterogeneous environment, etc. The scope of the theory notwithstanding, even the simplest possible formulation of a multi locus mutation--selection model in the case of an asexually reproduced haploid population still presents mathematical challenges. The goal of our manuscript is to introduce an approach that allows, at least in some special cases, a relatively straightforward derivation of the steady state distribution for such model.

Consider a haploid one locus asexually reproduced population with $N+1$ alleles. Let $p_i=p_i(t)$ be the frequency of the $i$-th allele at time $t$, $i=0,\ldots,N$; the corresponding Malthusian fitness is denoted $m_i$. Let $\mu_{ij}$ denote the mutation rate of allele $j$ to allele $i$. Then, assuming that the reproduction events and mutations are separated, we end up with a nonlinear system of ordinary differential equations (e.g., \cite{baake1999,crow1970introduction}) of the form
\begin{equation}\label{eq1:1}
    \dot p_i=(m_i-\overline{m})p_i+\sum_{j=0}^N \mu_{ij}p_j,\quad i=0,\ldots, N,
\end{equation}
where $\mu_{ii}=-\sum_{j\neq i}\mu_{ij}$, and $\overline{m}=\sum_{j=0}^Nm_jp_j$ is the mean population fitness. In the matrix form system \eqref{eq1:1} reads
\begin{equation}\label{eq1:2}
    \bs{\dot p}=(\bs M-\overline{m}\bs I)\bs p+\bs{\mathcal M}\bs p,
\end{equation}
where $\bs M=\diag(m_0,\ldots, m_N)$ is the fitness landscape, $\bs{\mathcal M}=(\mu_{ij})$ is the mutation matrix, $\bs p=(p_0,\ldots,p_N)^\top$, and $\bs I$ is the identity matrix. Model \eqref{eq1:2} is invariant with respect to rescalings of the Malthusian fitnesses as $\tilde m_i=m_i+\tilde m$ for any constant $\tilde m$ and does not allow for lethal mutations.

A basic fact about model \eqref{eq1:2} is as follows. Given that the matrix $\bs{\mathcal M}$ is irreducible, then there exists a unique globally stable positive selection--mutation equilibrium $\lim_{t\to\infty}\bs p(t)=\bs{\hat{p}}$ that can be found as the normalized eigenvector of the matrix $\bs M+\bs{\mathcal M}$ corresponding to the strictly dominant eigenvalue $\lambda=\hat{\overline{m}}=\sum_{j=0}^N m_j\hat p_j$ (e.g., \cite{burger2000mathematical,thompson1974eigen}). In Eigen's theory of the origin of life, which is equivalent to the selection--mutation approach in haploid populations \cite{wilke2005quasispecies}, the eigenvector $\bs{\hat p}$, which describes the equilibrium frequencies of the self-replicating macromolecules, was called the quasispecies \cite{eigen1971sma,eigen1988mqs,Eigen1977}.

To obtain further theoretical results on the form of the dominant eigenvalue $\lambda$ and/or selection--mutation equilibrium $\bs{\hat{p}}$, additional assumptions on the form of the fitness landscape $\bs M$ and/or mutation matrix $\bs{\mathcal M}$ are necessary. For example, a profitable way is to neglect the reverse mutations, i.e., put (e.g., \cite{wiehe1997model})
\begin{equation*}
    \begin{split}
       \mu_{ij} & >0,\quad i>j, \\
       \mu_{ij} & =0,\quad i<j,
     \end{split}
\end{equation*}
or even more restrictive (e.g., \cite{wagner1993difference})
\begin{equation*}
    \begin{split}
       \mu_{j+1,j} & >0, \\
       \mu_{ij} & =0,\quad i\neq j,\,i\neq j+1.
     \end{split}
\end{equation*}
A less restrictive assumption is to assume that allele $j$ can mutate only to the neighbors $j-1$ and $j+1$, other mutations are prohibited \cite{baake1999,baake2001mutation,Hermisson2002}. In the last case, assuming additionally that
\begin{equation}\label{eq1:3}
    \begin{split}
       \mu_{j-1,j} & =\mu j, \\
       \mu_{j+1,j} & =\mu(N-j),\\
       \mu_{jj}&=-\mu N,
     \end{split}
\end{equation}
for some constant $\mu>0$, it is possible to write down an explicit solution for the equilibrium frequencies $\hat p_i$ in the case of an additive or Fujiyama fitness landscape, defined as $m_j=-Mj$ for some constant $M>0$ \cite{baake2001mutation,higgs1994error,rumschitzki1987spectral} (we reproduce this solution, using our method, below in Example \ref{ex:5}).

We stress that no simple analytical expressions for the components of $\bs{\hat{p}}$ are known for a general fitness landscape $\bs M$ and the mutation scheme \eqref{eq1:3}; even in the case of a single or sharply peaked landscape, defined as $\bs M=\diag(m_0,0,\ldots,0),\,m_0>0$, there exists no analytical solution.

The mutation scheme \eqref{eq1:3} naturally arises if one changes the point of view from a one locus $N+1$ allele population to a biallelic $N$ locus haploid population, which is scrutinized in the quasispecies theory \cite{baake1999,bratus2013linear,eigen1988mqs,jainkrug2007}. To this end, consider a population of sequences, each sequence consists of $N$ sites (loci), and each site can be either in 0 or 1 state (two alleles per each locus). Let $\mu>0$ denote the mutation rate per site per sequence per time unit, such that 0s mutates to 1s and 1s mutates to 0s with the same rate $\mu>0$. Also assume that fitness is determined by the number of 1s in the sequence such that we have $N+1$ different sequence classes with fitnesses $m_0,m_1,\ldots,m_N$ (this is sometimes called permutation invariant or symmetric fitness landscape). Then the dynamics of the frequencies of classes are determined by model \eqref{eq1:2} with \eqref{eq1:3}, which is often called in the literature the paramuse or Crow--Kimura quasispecies model~\cite{baake1999}.

It was shown that model \eqref{eq1:2}, \eqref{eq1:3} is equivalent to a so-called Ising quantum chain (e.g., \cite{baake2001mutation}). This fact allowed obtaining a number of analytical results about the mean population fitness $\lambda=\hat{\overline{m}}$ (and for some other population averages) in the selection--mutation equilibrium when the sequence length approaches infinity ($N\to\infty$) under an appropriate scaling of the model parameters \cite{baake2001mutation}.

A similar infinite sequence point of view was taken in \cite{Hermisson2002} (see also \cite{Baake2007} for   a more recent generalization), where a maximum principle for the mean population fitness was formulated. For our needs the maximum principle can be formulated as follows (we note that a more general case is treated in \cite{Hermisson2002}). Assume that $m_i=Nr_i=Nr(x_i),\,x_i=i/N\in[0,1]$ and define $g(x)=\mu\bigl(1-2\sqrt{x(1-x)}\bigr)$. Then the scaled equilibrium fitness $\hat{\overline{r}}=\hat{\overline{m}}/N$ is given by
\begin{equation}\label{eq1:4}
    \hat{\overline{r}}\approx \hat{\overline{r}}_\infty=\sup_{x\in[0,1]}\bigl(r(x)-g(x)\bigr).
\end{equation}
Using the approximate expression for the mean fitness in \eqref{eq1:4} it is possible to obtain expressions for other averages such as the variance per site of fitness and of distance from the fittest class~\cite{Hermisson2002}. However, no attempt was made in \cite{Hermisson2002} to obtain analytical expressions for $\bs{\hat{p}}_\infty$ (we use index $\infty$ throughout the text to denote the expressions in the infinite sequence limit).

An exact integral representation of $\bs{\hat{p}}_\infty$ for the infinite sequence length for the mutation scheme \eqref{eq1:3} and the single peaked landscape was written in \cite{galluccio1997exact}, however, transparent analytical expressions were obtained only for $\hat{\overline{r}}_\infty$ and $\hat{p}_{\infty,0}$. In \cite{saakian2004solvable} a full solution for $\hat p_{\infty,i}$ was written down by disregarding the reverse mutations from class $j$ to class $j-1$. The same solution was rigorously obtained in \cite{bratus2013linear} together with estimates of the speed of convergence. In \cite{saakian2007new}, using the Hamilton--Jacobi formalism, a general solution for $\bs{\hat{p}}_\infty$ depending on an arbitrary scaled fitness landscape $r(x)$ is suggested in an integral form. This general solution is, however, not straightforward to apply to obtain relatively simple analytical expressions for the selection--mutation equilibrium frequencies $\hat{p}_{\infty,i}$. Therefore, we conclude that there exists no simple general way to find the quasispecies distribution $\bs{\hat{p}}_\infty$ even under the simplifying assumption of the infinite sequence length.

The goal of the present text is to suggest a straightforward way of calculating the selection--mutation equilibrium for the model \eqref{eq1:2}, \eqref{eq1:3} in the infinite sequence length limit that leads to transparent analytical expressions at least for some particular fitness landscapes (for an extensive background for the current work we refer to \cite{bratus2013linear} and \cite{semenov2014}).

\section{A general approach to solve for the selection--mutation equilibrium}\label{sec:2}
Our goal is to find approximations for the dominant eigenvalue $\lambda=\hat{\overline{m}}$ and the corresponding normalized eigenvector $\bs{\hat{p}}$ (such that $\sum_{i=0}^N\hat p_i=1$ and $\hat{p}_i\geq 0$ for any $i$) of the eigenvalue problem
\begin{equation}\label{eq2:1}
    (\bs M+\mu \bs Q)\bs{\hat{p}}=\lambda \bs{\hat{p}},
\end{equation}
where $\bs M=\diag(m_0,\ldots,m_N)$, $\mu>0$, and
$$
\bs Q=\begin{bmatrix}
               -N & 1 & 0 & 0 & \ldots & \ldots & 0 \\
               N & -N & 2 & 0 & \ldots & \ldots & 0 \\
               0 & N-1 & -N & 3 & \ldots & \ldots & 0 \\
               0 & 0 & N-2 & -N & \ldots & \ldots & 0 \\
               \ldots & \ldots & \ldots & \ldots & \ldots & \ldots & \ldots \\
               0 & 0 & \ldots & \ldots& 2 & -N & N \\
               0 & 0 & \ldots & \ldots & 0& 1 & -N \\
             \end{bmatrix}.
$$
We note that at the equilibrium both $\hat{\overline{m}}$ and $\bs{\hat{p}}$ are the functions of the mutation rate $\mu$: $\hat{\overline{m}}=\hat{\overline{m}}(\mu),\,\bs{\hat{p}}=\bs{\hat{p}}(\mu)$. Together with the matrix $\bs Q$ we also consider a linear differential operator
\begin{equation}\label{eq2:2}
    \mathcal Q\colon P(s)\longrightarrow (1-s^2)P'(s)-N(1-s)P(s),
\end{equation}
acting on the $(N+1)$-dimensional vector space of polynomials of degree less or equal $N$. By direct calculations, matrix $\bs Q$ is the matrix of $\mathcal Q$ in the standard basis $\{1,s,\ldots,s^N\}$.  For any $\bs m=(m_0,\ldots,m_N)\in \R^{N+1}$ and $P(s)=\sum_{i=0}^N p_is^i$ we introduce the notation
$$
\bs m\circ P(s)=\sum_{i=0}^N m_ip_is^i.
$$
Then problem \eqref{eq2:1} can be rewritten for the unknown probability generating function $P(s)$ as
\begin{equation}\label{eq2:3}
    \bs m\circ P(s)+\mu\mathcal QP(s)=\hat{\overline{m}}\, P(s),
\end{equation}
where $\hat{\overline{m}}=\bs m\circ P(1)$, or
\begin{equation}\label{eq2:4}
    \bs m\circ P(s)+\mu(1-s^2)P'(s)-\mu N(1-s)P(s)=\hat{\overline{m}}\, P(s),
\end{equation}
with the normalization condition $P(1)=1$. Inasmuch as problem  \eqref{eq2:4} is equivalent to \eqref{eq2:1} then, due to the Perron--Frobenius theorem, there exists a unique solution $P(s)$ satisfying $P(1)=1$. There is little hope to be able to solve equation \eqref{eq2:4} explicitly (one such example, well known in the literature, is given below, see Example \ref{ex:5}). It is possible, however, to find approximations of the quantities of interest in the case $N\to\infty$ at least for some particular fitness landscapes under some additional assumptions.

To formulate the general approach, we introduce the notations
$$
\bs r=\frac{\bs m}{N}\,,\quad \hat{\overline{r}}=\frac{\hat{\overline m}}{N}\,.
$$
After dividing by $N$ equation \eqref{eq2:4} takes the form
$$
\bs r\circ P(s)+\frac{\mu}{N}(1-s^2)P'(s)-\mu(1-s)P(s)=\hat{\overline{r}}\, P(s).
$$

Now we make the following assumptions:
\begin{itemize}
\item[$\mathcal H1$:] There exists the limit
$$
\lim_{N\to\infty} P(s)=P_\infty(s).
$$
The distribution $\hat{p}_{\infty,i}=\hat{p}_i,\,i=0,1,2,\ldots$, corresponding to $P_\infty(s)=\sum_{i=0}^\infty \hat p_is^i$, will be called the limit distribution.
\item[$\mathcal H2$:]
$$
\lim_{N\to\infty} \frac{\mu}{N}(1-s^2)P'(s)=0.
$$
\item[$\mathcal H3$:] For some limit operator $\bs r_\infty$
$$
\lim_{N\to\infty} \bs r\circ P(s)=\bs r_\infty\circ P_\infty(s).
$$
\end{itemize}
\begin{remark} The assumption $\mathcal H2$ is a formal consequence of $\mathcal H1$ and $\mathcal H3$, but we decided to keep it in the list because it gives the main idea of the suggested method.
\end{remark}

If $\mathcal H1$--$\mathcal H3$ hold then problem \eqref{eq2:4} is reduced to a nonlinear functional equation with respect to the unknown probability generating function $P_\infty(s)$,
\begin{equation}\label{eq2:5}
    -\mu(1-s)P_\infty(s)+\bs r_\infty\circ P_\infty(s)=\hat{\overline{r}}_\infty P_\infty(s),
\end{equation}
with the conditions
\begin{equation}\label{eq2:6}
    \bs{r}_\infty\circ P(1)=\hat{\overline{r}}_\infty,\quad P_\infty(1)=1.
\end{equation}

Problem \eqref{eq2:5}--\eqref{eq2:6} can be effectively solved at least for some simple choices of the fitness landscape $\bs M$ (see the next section for representative examples). A careful scrutiny of validity of the obtained solutions requires a deeper analysis of the convergence of the eigenvalue $\hat{\overline{m}}$ and the quasispecies $\bs{\hat{p}}$ when $N\to\infty$. This, for instance, can be done  with the help of parametric solutions to \eqref{eq2:1} introduced in \cite{bratus2013linear} (in Appendix \ref{app:1} we outline the approach used in \cite{bratus2013linear}, in Appendix \ref{app:2} we provide an alternative parametric solution approach, which is illustrated by applying it to Example \ref{ex:7}). Notwithstanding these concerns, a formal solution of \eqref{eq2:5}--\eqref{eq2:6} is of significant value, because, as the examples show, the found solutions closely approximate, even for moderate values of $N$, numerical solutions of \eqref{eq2:1}.

In a sense our approach is a generalization of the so-called random variable technique (e.g., \cite{bailey1990elements}), and assumption $\mathcal H2$ implies that we disregard all the mutations from $j$ to $j-1$ classes, similarly to \cite{wagner1993difference,wiehe1997model}. To see this, consider the mutation scheme of the form
\begin{equation}\label{eq2:7}
    \begin{split}
       \mu_{j-1,j} & =\mu_1 j, \\
       \mu_{j+1,j} & =\mu_2(N-j),\\
       \mu_{jj}&=-\mu_1j-\mu_2(N-j),
     \end{split}
\end{equation}
where $\mu_1\neq \mu_2$. Then, as it can be directly checked, operator $\mathcal Q$ takes the form
$$
\mathcal Q\colon P(s)\longrightarrow \mu_2(s-1)\left(N-s\frac{\D}{\D s}\right)P(s)+\mu_1 (1-s)P'(s).
$$
After dividing by $N$ and formally taking the limit, the only term that is left is
$$
-\mu_2(1-s)P(s),
$$
which agrees with \eqref{eq2:5} and shows that for the limit equation the rate of backward mutations $\mu_1$ is neglected.

To conclude this section, we suggest the following approach: Solve problem \eqref{eq2:5}--\eqref{eq2:6} for each $\mu$. For the found solution $P_\infty(s)$ check $\mathcal H1$--$\mathcal H3$. If the hypotheses do not hold then notice that $P_\infty(s)\equiv 0$ solves \eqref{eq2:5}--\eqref{eq2:6}, with $\hat{\overline{r}}_\infty(\mu)\equiv 0$. This corresponds to the delocalization phenomenon of the quasispecies distribution, or, in terms of the Ising model, the phase transition, which was called the error threshold in the quasispecies theory. A number of examples illustrating this approach are given in the following section.

\section{Examples of the steady state distributions}
Here we present several examples of the selection--mutation equilibrium for known and new fitness landscapes; we also compare the analytical results with the numerical calculations. Two well known cases are Example \ref{ex:1}, treated in \cite{bratus2013linear,galluccio1997exact,saakian2004solvable}, and Example \ref{ex:5}, treated originally in \cite{baake2001mutation,higgs1994error,rumschitzki1987spectral}. We present full solutions in these two cases to demonstrate how our approach works. Example \ref{ex:7} was discussed and partially analyzed in \cite{bratus2013linear}, however, no derivation for the steady state distribution $\bs{\hat{p}}_\infty$ was provided; here we present all the details. Other examples are new and are not treated anywhere else to the best of our knowledge.

\begin{example}[Single peaked landscape]\label{ex:1} We start with a testbed (both numerical \cite{swetina1982self} and analytical \cite{galluccio1997exact}) for the quasispecies model, which was called the single or sharply peaked landscape.

Let
$$
\bs r_\infty=(1,0,\ldots,0,\ldots).
$$
Then
$$
\bs r_\infty \circ P_\infty(s)=\sum_{i=0}^\infty r_{i}\hat p_i t^i=\hat{p}_0=P_\infty(0).
$$
Equation \eqref{eq2:5} takes the form
$$
-\mu(1-s)P_\infty(s)+P_\infty(0)=\hat{\overline{r}}_\infty P_\infty(s).
$$
Plug $s=0$ in the last expression and find
$$
-\mu P_\infty(0)+P_\infty(0)=\hat{\overline{r}}_\infty P_\infty(0).
$$
Assuming that $P_\infty(0)\neq 0$ we find
$$
\hat{\overline{r}}_\infty(\mu)=1-\mu,
$$
and hence, using the condition $P_\infty(1)=1$,
$$
P_\infty(s)=\frac{1-\mu}{1-\mu s}\,,
$$
which is the probability generating function of the geometric distribution with the parameter $\mu$. Therefore the limit distribution is geometric
$$
\hat{p}_{\infty,i}=(1-\mu)\mu^i,\quad i=0,1,\ldots.
$$

We see, in view of $\mathcal H1$---$\mathcal H3$, that the discussion above holds only for $\mu<1$, therefore at $\mu=1$ the structure of the limit distribution abruptly changes and we obtain the solution $P_\infty(s)=0$ for $\mu\geq 1$. The mean population fitness has the form shown in Fig. \ref{fig:1}a. This abrupt change in the quasispecies distribution was called by Eigen et al. the error threshold \cite{biebricher2005error} (see also \cite{bratus2013linear,Hermisson2002} for an extensive discussion of this notion).
\begin{figure}
\centering
\includegraphics[width=0.95\textwidth]{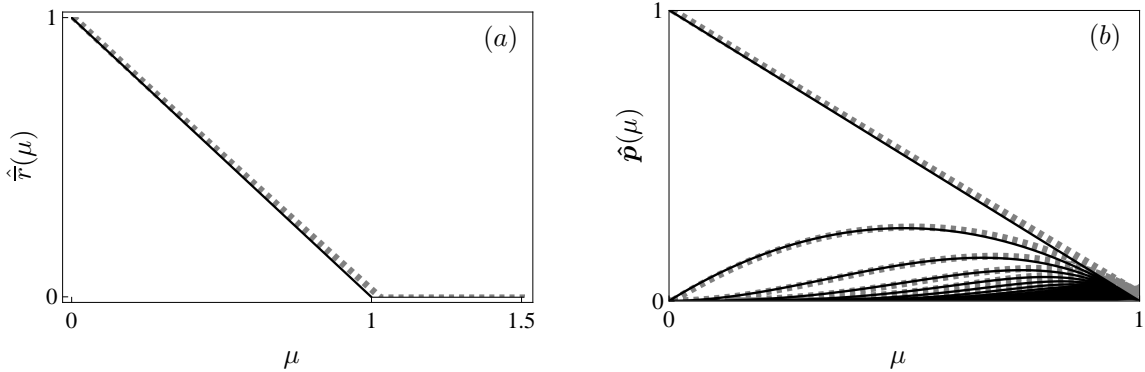}
\caption{Comparison of numerical calculations for the single peaked landscape $\bs r=(1,0,\ldots,0)$ with $N=50$ with the theoretical predictions of Example \ref{ex:1}. The black solid lines are the exact solutions for the case $N\to\infty$ and the grey dashed lines are numerical computations. $(a)$ The mean population fitness versus the mutation rate. $(b)$ The selection--mutation equilibrium $\bs{\hat{p}}$ versus the mutation rate. After $\mu\geq 1$ the quasispecies distribution becomes degenerate (binomial)}\label{fig:1}
\end{figure}
\end{example}

\begin{example}\label{ex:2}As a second example, consider a slight generalization of the single peaked landscape in the form
$$
\bs r_\infty=(2,1,0,\ldots,0,\ldots).
$$
Then $\bs r_\infty \circ P_\infty(s)=2P_\infty (0)+P'_\infty(0)s,$
and \eqref{eq2:5} takes the form
$$
-\mu(1-s)P_\infty(s)+2P_\infty(0)+P_\infty'(0)s=\hat{\overline{r}}_\infty P(s).
$$
Plugging $s=0$ yields
$$
-\mu P_\infty(0)+2P_\infty(0)=\hat{\overline{r}}_\infty P_\infty(0).
$$
Assuming $P_\infty(0)\neq 0$ we find
$$
\hat{\overline{r}}_\infty=2-\mu,
$$
and hence
$$
-\mu(1-s)P_\infty(s)+2 P_\infty(0)+P_\infty'(0)s=(2-\mu)P_\infty(s),
$$
or
\begin{equation}\label{eq3:1}
\mu s P_\infty(s) +2P_\infty(0)+P_\infty'(0)=2P_\infty(s).
\end{equation}
After differentiating the last expression with respect to $s$ and plugging $s=0$ we find $\mu P_\infty(0)=P_\infty'(0).$
Plugging this into \eqref{eq3:1} implies
$$
P_\infty(s)=P_\infty(0)\frac{2+\mu s}{2-\mu s}\,,
$$
and finally, using the condition $P_\infty(1)=1$, we obtain
$$
P_\infty(s)=\frac{(2-\mu)(2+\mu s)}{(2+\mu)(2-\mu s)}=\frac{2-\mu}{2+\mu}+\sum_{j=1}^\infty\frac{(2-\mu)\mu^j}{(2+\mu)2^{j-1}}\,s^j.
$$
Again, the reasonings above work only for $\mu<2$, and for $\mu\geq 2$ we obtain that $\hat{\overline{r}}(\mu)=0$ and the quasispecies distribution is degenerate (see also Fig. \ref{fig:2}).
\begin{figure}
\centering
\includegraphics[width=0.95\textwidth]{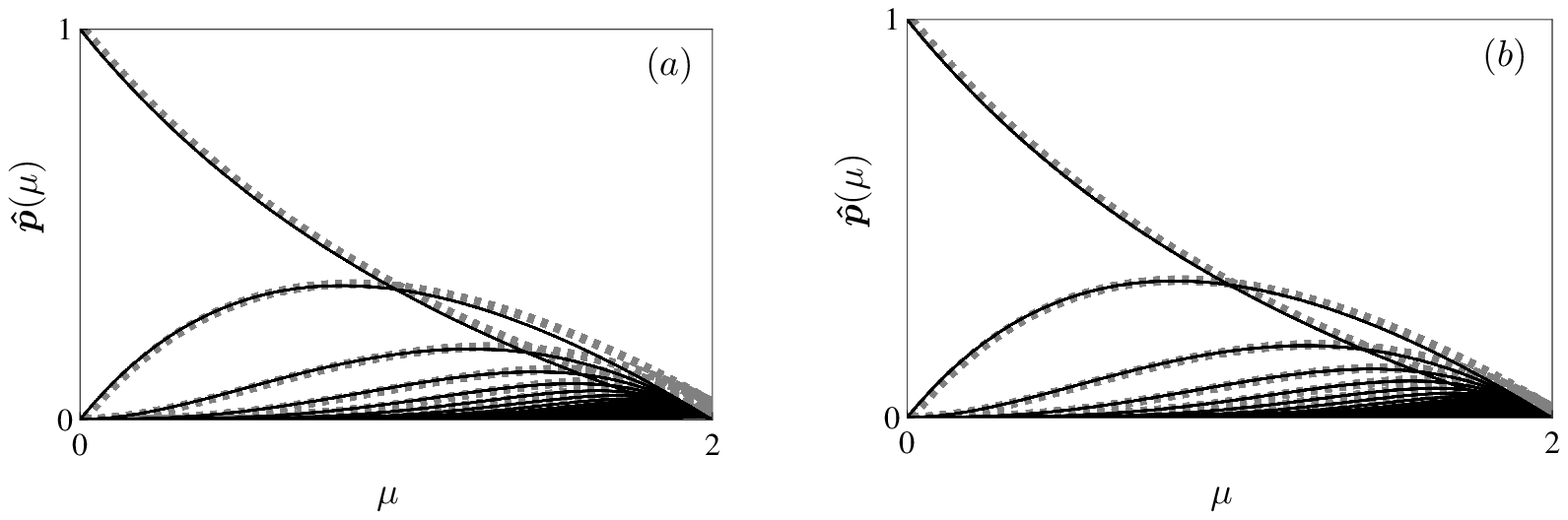}
\caption{Comparison of numerical calculations for the fitness landscape in Example \ref{ex:2} with the theoretical predictions. The black solid lines are the exact solutions for the case $N\to\infty$ and the grey dashed lines are numerical computations. $(a)$ $N=50$; $(b)$ $N=100$}\label{fig:2}
\end{figure}
\end{example}

\begin{example}\label{ex:3}Let
$$
\bs r_\infty=(1,2,0,\ldots,0,\ldots).
$$
Then
$
\bs r_\infty\circ P_\infty(s)=P_\infty(0)+2P_\infty'(0)s,
$ and \eqref{eq2:5} takes the form
$$
-\mu(1-s)P_\infty(s)+P_\infty(0)+2P_\infty'(0)s=\hat{\overline{r}}_\infty P_\infty(s).
$$
We must assume that $P_\infty(0)=\hat p_{\infty,0}=0$ since $\max r_k=2=r_1>r_0$, therefore we cannot just plug $s=0$ in the last equality. Instead, we differentiate it and find
$$
-\mu(1-s)P_\infty'(s)+\mu P_\infty(s)+2P_\infty'(0)=\hat{\overline{r}}_\infty P_\infty'(s).
$$
Then, for $s=0$, assuming that $\hat p_{\infty,0}=0$ and $P_\infty'(0)\neq 0$,
$$
\hat{\overline{r}}_\infty=2-\mu.
$$
Therefore,
$$
-\mu(1-s)P_\infty(s)+2P_\infty'(0)s=(2-\mu)P_\infty(s),
$$
or
$$
P_\infty(s)=\frac{2P_\infty'(0)s}{2-\mu s}\,,
$$
which, together with $P_\infty(1)=1$, gives
$$
P_\infty(s)=\frac{(2-\mu)s}{2-\mu s}=\sum_{j=1}^\infty\left(1-\frac{\mu}{2}\right)\frac{\mu^{j-1}}{2^{j-1}}\,s^j,
$$
which holds only for $\mu<2$, for $\mu\geq 2$ the distribution becomes degenerate (see Fig. \ref{fig:3}).
\begin{figure}
\centering
\includegraphics[width=0.95\textwidth]{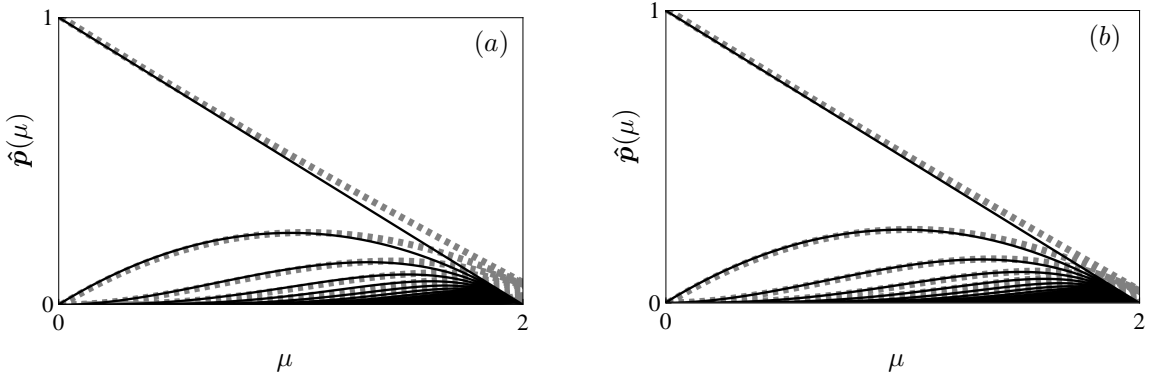}
\caption{Comparison of numerical calculations for the fitness landscape in Example \ref{ex:3} with the theoretical predictions. The black solid lines are the exact solutions for the case $N\to\infty$ and the grey dashed lines are numerical computations. $(a)$ $N=50$; $(b)$ $N=100$}\label{fig:3}
\end{figure}
\end{example}
\begin{remark}\label{rm:3:4}In general, if $\max r_k=r_a$ and $r_j<r_a$ for $j>a$ then we need additional initial conditions
$$
P_\infty(0)=P_\infty'(0)=\ldots=P_\infty^{(a-1)}(0)=0.
$$
These initial conditions are motivated by the comparison of the theoretical predictions with the numerical calculations, and at this point we lack an analytical proof of the validity of these conditions in general.
\end{remark}
\begin{example}[A geometric landscape]\label{ex:4}
Let $0<q<1$ and
$$
\bs r_\infty=(1,q,q^2,\ldots,q^N,\ldots),
$$
then $
\bs r_\infty \circ P_\infty(s)=P_\infty(qs),
$ and \eqref{eq2:5} reads
$$
-\mu(1-s)P_\infty(s)+P_\infty(qs)=\hat{\overline{r}}_\infty P_\infty(s).
$$
Plugging $s=0$ and assuming $P_\infty(0)\neq 0$ we find
$$
\hat{\overline{r}}_\infty=1-\mu,
$$
which implies
\begin{equation}\label{eq3:2}
P_\infty(qs)=(1-\mu s)P_\infty (s).
\end{equation}
Using the fact $P_\infty(1)=1$ and plugging into the last expression $s=1,\,s=q,\,s=q^2,\ldots$ we find
$$
P_\infty(q)=1-\mu,\, P_\infty(q^2)=(1-\mu q)P_\infty(q)=(1-\mu)(1-\mu q),\,\ldots\,,P_\infty(q^n)=\prod_{j=0}^{n-1}(1-\mu q^j),\ldots
$$
Taking the limit $n\to\infty$ implies
$$
P_\infty(0)=\lim_{n\to\infty} P_\infty(q^n)=\prod_{j=0}^\infty (1-\mu q^j).
$$
Assuming $P_\infty(s)=\sum_{j=0}^\infty \hat{p}_j s^j$ in \eqref{eq3:2}, we obtain
$$
q^n\hat p_n=\hat p_n-\mu \hat p_{n-1},\quad \text{or}\quad \hat p_n=\frac{\mu}{1-q^n}\,\hat p_{n-1},
$$
which gives, for $0<\mu<1$, the limit distribution (see also Fig. \ref{fig:4})
$$
\hat p_0=\prod_{j=0}^\infty (1-\mu q^j),\quad \hat p_n=\frac{\mu^n}{\prod_{k=1}^n(1-q^k)}\,\hat p_0.
$$The value $\mu=1$ is critical and corresponds to the error threshold. See Fig. \ref{fig:4} for comparison of the theoretical predictions and numerical computations.

We note that there are effective methods to calculate the expressions of the form $\prod_{j=0}^\infty (1-\mu q^j)$ numerically. For instance, in \textit{Mathematica}$^\copyright$ this is done with the help of function \verb"QPochhammer["$\mu,q$\verb"]".
\begin{figure}
\centering
\includegraphics[width=0.95\textwidth]{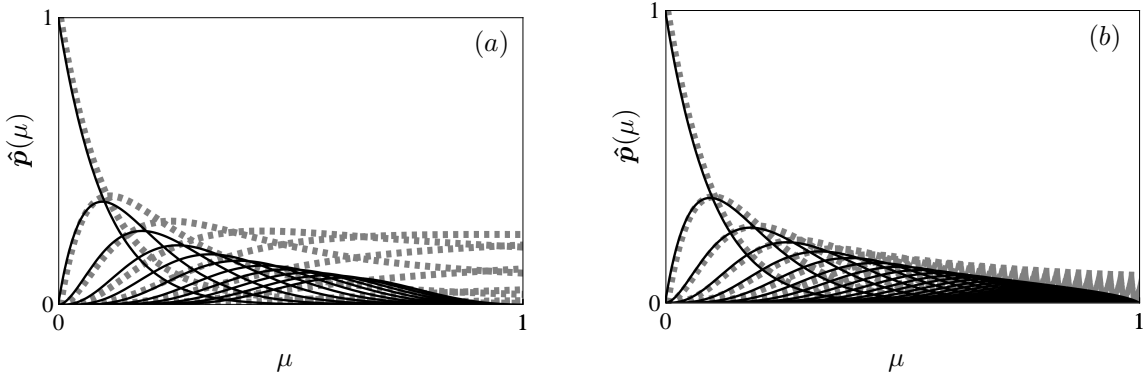}
\caption{Comparison of numerical calculations for the fitness landscape in Example \ref{ex:4} with the theoretical predictions. The black solid lines are the exact solutions for the case $N\to\infty$ and the grey dashed lines are numerical computations. $(a)$ $N=10$; $(b)$ $N=50$}\label{fig:4}
\end{figure}
\end{example}

\begin{example}[Additive or Fujiyama fitness landscape]\label{ex:5} The only fitness landscape for which problem \eqref{eq1:2} with the mutation scheme \eqref{eq1:3} can be analytically solved for the selection--mutation equilibrium is the additive fitness landscape, which we define here as
$$
\bs r=\bs r_N=\left(1,1-\frac 1N,1-\frac 2N,\ldots,1-\frac NN=0\right).
$$
The solution can be found in \cite{baake2001mutation,higgs1994error} and arguably is most naturally derived using the tensor products for the representation of matrices $\bs M+\bs{\mathcal M}$. Here we re-derive the same solution using the generating function approach.

Since
$$
\bs r_N\circ P_N(s)=P_N(s)-\frac{s}{N}P'_N(s),
$$
then \eqref{eq2:4} reads
$$
\frac{\mu}{N}(1-s^2)P'_N(s)-\mu(1-s)P_N(s)+P_N(s)-\frac{s}{N}P_N'(s)=\hat{\overline{r}}_NP_N(s).
$$
Making the substitution $P_N(s)=W^N(s)$ yields the ODE
$$
\mu(1-s^2)W'(s)-\mu(1-s)W(s)+W(s)-sW'(s)=\hat{\overline{r}}_NW(s),
$$
or
$$
\frac{W'(s)}{W(s)}=\frac{A}{s+a}+\frac{B}{s+b}\,,
$$
where
\begin{equation*}
    \begin{split}
      a &=\frac{\sqrt{1+4\mu^2}+1}{2\mu}\,,\quad b=\frac{\sqrt{1+4\mu^2}-1}{2\mu}\,,\quad ab=1,\quad a>0,\quad 0<b<1,  \\
      A & =\frac 12+\frac{2\hat{\overline{r}}_N+2\mu-1}{2\sqrt{1+4\mu^2}}\,,\quad B  =\frac 12-\frac{2\hat{\overline{r}}_N+2\mu-1}{2\sqrt{1+4\mu^2}}\,,\quad A+B=1.
    \end{split}
\end{equation*}
Integrating this ODE yields
$$
W(s)=C(s+a)^A(s-b)^B,
$$
with the condition $C(1+a)^A(1-b)^B=1$. We formally have that $W(b)=0$, but this cannot occur, since $0<b<1$ and the polynomial $P_N(s)$ has all non-negative coefficients and is not equal to zero anywhere on the interval $[0,1]$. This implies that $B=0$, and therefore $A=1$, and
$$
W(s)=\frac{s+a}{1+a}\,.
$$
Moreover, the condition $B=0$ implies that independently of $N$
$$
\hat{\overline{r}}_N=\frac{1-2\mu+\sqrt{1+4\mu^2}}{2}\,.
$$
The final solution is
$$
P_N(s)=\left(\frac{s+a}{1+a}\right)^N,
$$
and therefore the steady state distribution is binomial:
$$
\hat p_{N,j}=\binom{N}{j}\frac{b^j}{(1+b)^N}\,\quad j=0,\ldots,N,
$$
which holds for any $\mu>0$, there exists no error threshold for this fitness landscape.
\end{example}

\begin{example}\label{ex:6} Consider a close relative of the additive fitness landscape in the form
$$
\bs r_\infty=\left(1,1-\frac 1K,1-\frac 2K,\ldots,1-\frac KK=0,\ldots\right),
$$
where $K\in\N$ and does not depend on $N$.

We have
$$
\bs r_\infty\circ P_\infty(s)=\sum_{a=0}^{K-1}\left(1-\frac{a}{K}\right)\hat p_a s^a,
$$
and \eqref{eq2:5} reads
$$
-\mu(1-s)P_\infty(s)+\sum_{a=0}^{K-1}\left(1-\frac{a}{K}\right)\hat p_a s^a=\hat{\overline{r}}_\infty P_\infty(s).
$$
Plugging $s=0$ and assuming $P_\infty(0)\neq 0$ implies
$$
\hat{\overline{r}}_\infty =1-\mu.
$$
Therefore,
\begin{equation}\label{eq3:3}
P_\infty(s)=\frac{\sum_{a=0}^{K-1}\left(1-\frac{a}{K}\right)\hat p_a s^a}{1-\mu s}\,,
\end{equation}
and we need to determine $\hat p_a$ for $a=0,\ldots,K-1$. Since we have
$$
\sum_{a=0}^{K-1}\left(1-\frac{a}{K}\right)\hat p_a s^a=(1-\mu s)\sum_{a=0}^\infty \hat p_as^a,
$$
then for $1\leq a\leq K-1$
$$
\left(1-\frac aK\right)\hat p_a=\hat p_a-\mu\hat p_{a-1},
$$
or
$$
\hat p_a=\frac{\mu K}{a}\hat p_{a-1}.
$$
The last recurrent formula implies that for $1\leq a\leq K-1$
$$
\hat p_a=\frac{(\mu K)^a}{a!}\hat p_0.
$$
To determine $\hat p_0$, we use $P_\infty(1)=1$, which yields, for $\mu<1$,
$$
\sum_{a=0}^{K-1}\left(1-\frac{a}{K}\right)\hat p_a=1-\mu,
$$
or, using the expressions for $\hat p_a$:
$$
1=\hat p_0\left(\sum_{a=0}^{K-2}\frac{(\mu K)^a}{a!}+\frac{(\mu K)^{K-1}}{(1-\mu)(K-1)!}\right),
$$
which allows us to find $\hat p_0$. Now we determined all $\hat p_a$ for $0\leq a\leq K-1$ and from \eqref{eq3:3} we have that for $j\geq K$
$$
\hat p_j=\sum_{a=0}^{K-1}\left(1-\frac{a}{K}\right)\hat p_a \mu^{j-a}.
$$
To simplify the expressions for $\hat p_j$ we note that due to the central limit theorem (CLT), for $K\to\infty$,
$$
\hat p_0\sim e^{-\mu K}.
$$
Indeed, if $\xi_1,\ldots \xi_K$ are independent identically distributed Poisson random variables with the mean $\textsf{E}{\xi}_i=\mu$, then $X_K=\xi_1+\ldots\xi_K$ has the Poisson distribution with the mean and the variance $\textsf{E}X_K=\textsf{Var}\,X_K=\mu K$. Therefore, using the CLT,
\begin{align*}
\sum_{a=0}^{K-2}\frac{(\mu K)^a}{a!}e^{-\mu K}&+\frac{(\mu K)^{K-1}}{(1-\mu)(K-1)!}e^{-\mu K}=\textsf{P}(X_K\leq K-2)+\frac{\textsf{P}(X_K=K-1)}{1-\mu}\\
&=\textsf{P}\left(\frac{X_K-\mu K}{\sqrt{\mu K}}\leq \frac{(1-\mu)K-2}{\sqrt{\mu K}}\right)+\frac{\textsf{P}(X_K=K-1)}{1-\mu}\to \textsf{P}(X<\infty)+0=1
\end{align*}
as $K\to\infty$. Here $X$ is the standard normally distributed random variable.

Therefore, for the case $1\ll K\ll N$ the equilibrium distribution is approximately Poisson with parameter $\mu K$ (see Fig. \ref{fig:5}).
\begin{figure}
\centering
\includegraphics[width=0.95\textwidth]{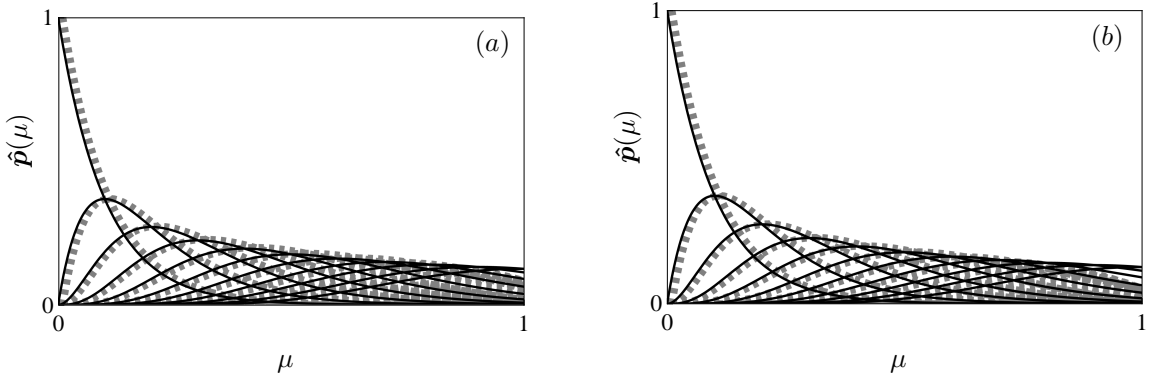}
\caption{Comparison of numerical calculations for the fitness landscape in Example \ref{ex:6} with the theoretical predictions. The black solid lines are the exact solutions for the case $N\to\infty$ and the grey dashed lines are numerical computations. $(a)$ $N=100,\,K=10$; $(b)$ $N=200,\,K=10$}\label{fig:5}
\end{figure}
\end{example}

In all the examples above the limit mean fitness $\hat{\overline{r}}_\infty$ can be determined from the maximum principle \eqref{eq1:4}. The next example shows that in the case when $r(x)$ has a point of discontinuity such that at this point function $r(x)$ neither left nor right continuous then a formal application of the maximum principle \eqref{eq1:4} may lead to incorrect conclusions.
\begin{example}\label{ex:7}
Let $N=2A$ be an even number, and
$$
\bs r_N=(0,\ldots,0,1,0,\ldots,0),
$$
where 1 is exactly at the $A$-th position. In \cite{bratus2013linear}, using a parametric solution to the eigenvalue problem (this solution is outlined in Appendix \ref{app:1}), it was proved that
$$
\hat{\overline{r}}_\infty=\sqrt{\mu^2+1}-\mu,
$$
which is defined for any $\mu>0$, there is no error threshold for this fitness ladscape. In Appendix \ref{app:2} we re-derive this result using a new approach.

The maximum principle \eqref{eq1:4} for this example cannot be applied because $r(x)$ is neither left nor right continuous at the point $x=0.5$. Formal application of the maximum principle leads to incorrect conclusion (e.g., for $\mu=1$ it predicts that $\hat{\overline{r}}\approx 1$, which is wrong, the exact value is $\sqrt{2}-1$). In \cite{bratus2013linear} also the expressions for the limit distribution $\bs{\hat{p}}_\infty$ were given without a full derivation. Here we show, using the method of generating functions, that the coordinates of the selection--mutation equilibrium indeed can be found in an explicit form.

In the following it will be convenient to consider the generating functions in the form of the Laurent series
$$
U(s)=\sum_{n=-\infty}^{\infty}u_ns^n.
$$
In \eqref{eq2:4} we make the substitution $P_N(s)=P_{2A}(s)=s^AU_A(s)$. Since the fitness landscape is symmetric, then the coefficients of $U_A(s)$ are also symmetric:
$$
U_A(s)=u_{A,0}+\sum_{n=1}^A u_{A,n}(s^n+s^{-n}),\quad \hat p_{A\pm n}=u_{A,n},\quad U_A(1)=1.
$$
After dividing by $2A$ equation \eqref{eq2:4} becomes
$$
\frac{\mu}{2A}(1-s^2)U'_A(s)+\frac{\mu}{2}(s^{-1}+s-2)U_A(s)+u_{A,0}=\hat{\overline{r}}_{2A}U_A(s).
$$
Similarly to $\mathcal H1$---$\mathcal H3$ we assume that, given that $A\to\infty$, the first term in the last equality vanishes, and $U_A(s)$ turns into
$$
U_\infty(s)=u_0+\sum_{n=1}^\infty u_n(s^n+s^{-n}),\quad U_\infty(1)=1=u_0+2\sum_{n=1}^\infty u_n.
$$
Therefore, we have the limit equation
$$
\frac{\mu}{2}(s^{-1}+s-2)U_\infty(s)+u_0=\hat{\overline{r}}_\infty U_\infty(s),\quad U_\infty(1)=1,
$$
or
$$
u_0=\left(\hat{\overline{r}}_\infty+\mu-\frac{\mu}{2}(s^{-1}+s)\right)\left(u_0+\sum_{n=1}^\infty(s^n+s^{-n})\right).
$$
Plugging in $s=1$ we find
$$
u_0=\hat{\overline{r}}_\infty.
$$
Moreover, by equating the coefficients at $s^{-n}+s^n$, we obtain the system
$$
u_0=(\hat{\overline{r}}_\infty+\mu)u_0-\mu u_1,\quad 0=(\hat{\overline{r}}_\infty +\mu)u_n-\frac{\mu}{2}(u_{n-1}+u_{n+1}),\quad n\geq 1.
$$
This system has the following solution, which can be directly checked,
$$
u_n=\hat{\overline{r}}_\infty \left(\frac{1-\hat{\overline{r}}_\infty}{1+\hat{\overline{r}}_\infty}\right)^n,\quad \hat{\overline{r}}_\infty=\sqrt{\mu^2+1}-\mu,\quad n=0,1,\ldots.
$$
Therefore, the limit distribution $\bs{\hat p}_\infty$ is two-sided geometric, and for large $N=2A$ we have approximately
$$
\hat p_{A\pm n}\approx u_n=\hat{\overline{r}}_\infty \left(\frac{1-\hat{\overline{r}}_\infty}{1+\hat{\overline{r}}_\infty}\right)^n.
$$
Numerical computations confirm this conclusion (see Fig. \ref{fig:6}).

\begin{figure}[!th]
\centering
\includegraphics[width=0.95\textwidth]{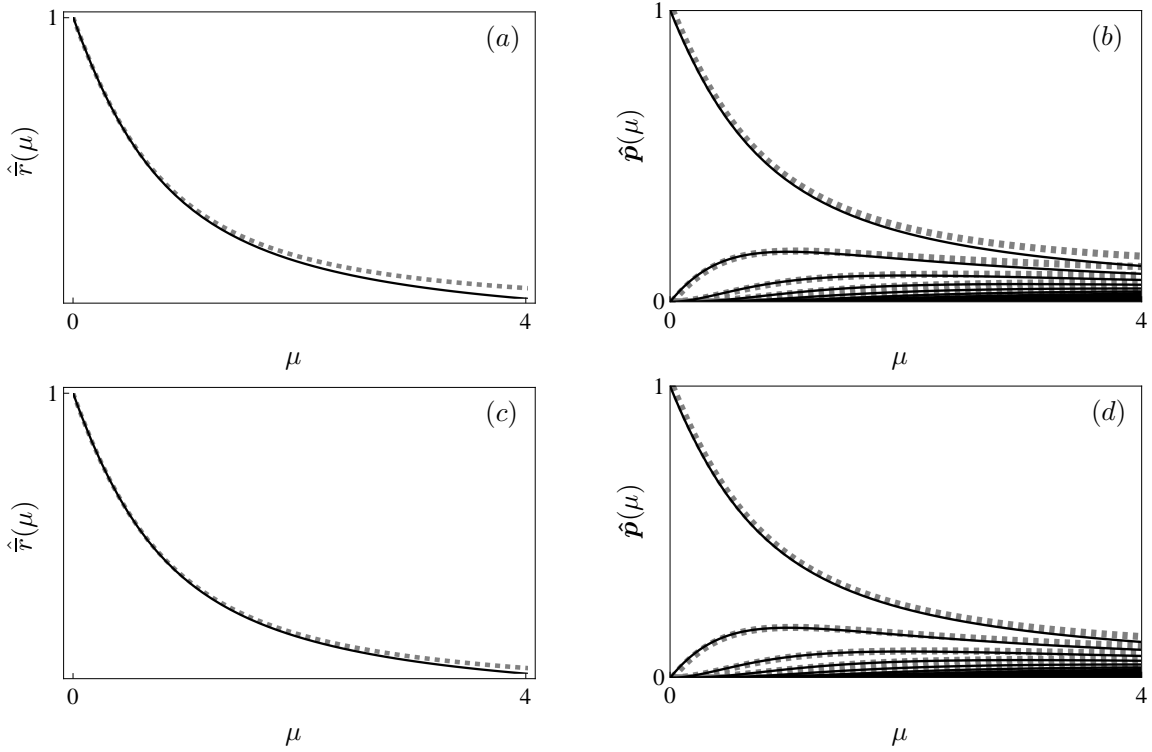}
\caption{Comparison of numerical calculations for the fitness landscape in Example \ref{ex:7} with the theoretical predictions. The black solid lines are the exact solutions for the case $N\to\infty$ and the grey dashed lines are numerical computations. $(a)$ and $(b)$ show the mean population fitness and the quasispecies distribution respectively for $N=100$ ; $(c)$ and $(d)$ show the same for $N=200$}\label{fig:6}
\end{figure}
\end{example}

\begin{example}[General formulas]\label{ex:8} As a final example, consider now a general fitness landscape, given by
\begin{equation}\label{eq3:4}
    \bs r_\infty=(r_0,r_1,\ldots),\quad r_0>r_i,\quad i=1,2,\ldots.
\end{equation}
We can prove
\begin{lemma}\label{lem:3:10} Suppose (dropping the subscript $\infty$ for notational convenience) that
$$
P(s)=\sum_{n=0}^\infty\hat{p}_ns^n
$$
gives a non-degenerate limit distribution
$$
\bs{\hat{p}}_\infty=(\hat{p}_0,\hat{p}_1,\ldots)
$$
such that $\hat{p}_0=P(0)>0$.

Then
\begin{equation}\label{eq3:5}
    \hat{\overline{r}}_\infty=r_0-\mu,\quad \hat p_n=\frac{\mu^n\hat{p}_0}{\prod_{j=1}^n(r_0-r_j)}\,,\quad n>0,\quad \hat{p}_0=\frac{1}{\displaystyle 1+\sum_{n=1}^\infty\frac{\mu^n}{\prod_{j=1}^n(r_0-r_j)}}\,.
\end{equation}
\end{lemma}
Expressions \eqref{eq3:5} generalize considered above Examples \ref{ex:1}--\ref{ex:3} and \ref{ex:4}.
\begin{proof} We rewrite \eqref{eq2:5} as follows:
$$
\mu P(s)=\frac{\hat{\overline{r}}_\infty P(s)-\bs{r}_\infty\circ P(s)}{s-1}=\frac{\hat{\overline{r}}_\infty (P(s)-1)-(\bs{r}_\infty\circ P(s)-\hat{\overline{r}}_\infty)}{s-1}\,,
$$
or, taking into account \eqref{eq2:6},
$$
\mu \sum_{n=0}^\infty \hat{p}_ns^n=\hat{\overline{r}}_\infty\sum_{n=1}^\infty \hat{p}_n\frac{s^n-1}{s-1}-\sum_{n=1}^\infty r_n\hat{p}_n\frac{s^n-1}{s-1}\,,
$$
or
\begin{equation}\label{eq3:6}
\mu \sum_{n=0}^\infty \hat{p}_ns^n=\hat{\overline{r}}_\infty\sum_{n=1}^\infty \hat{p}_n(1+s+\ldots+s^{n-1})-\sum_{n=1}^\infty r_n\hat{p}_n(1+s+\ldots+s^{n-1})\,.
\end{equation}
Substituting $s=0$ yields
$$
\mu\hat{p}_0=\hat{\overline{r}}_\infty\sum_{n=1}^\infty \hat{p}_n-\sum_{n=1}^\infty r_n\hat{p}_n=\hat{\overline{r}}_\infty(1-\hat{p}_0)-(\hat{\overline{r}}_\infty-r_0\hat{p}_0)=(r_0-\hat{\overline{r}}_\infty)\hat{p}_0.
$$
By assumption $\hat{p}_0>0$, and we get the first equality in \eqref{eq3:5}.

Now consider the coefficient at $s$ in \eqref{eq3:6}. We have
$$
\mu \hat{p}_1=\hat{\overline{r}}_\infty(1-\hat{p}_0-\hat{p}_1)-(\hat{\overline{r}}_\infty-r_0\hat{p}_0-r_1\hat{p}_1)=(r_0-\hat{\overline{r}}_\infty)\hat{p}_0-(\hat{\overline{r}}_\infty-r_1)\hat{p}_1,
$$
or, using the first inequality in \eqref{eq3:5},
$$
\hat{p}_1=\frac{\mu \hat{p}_0}{r_0-r_1}\,.
$$
Proceeding by induction on $n$ and comparing the coefficients at $s^n$ in \eqref{eq3:6}, we prove the second equality in \eqref{eq3:5}. The last equality follows from the condition $\sum_{n=0}^\infty\hat{p}_n=1$.
\end{proof}
\begin{remark}The expressions for the limit distribution in Lemma \ref{lem:3:10} can be generalized to the case
$$
\bs r_\infty=(r_0,r_1,\ldots),\quad r_k\geq r_i,\quad i=0,1,\ldots,k-1,\quad r_k>r_i,\quad i=k+1,\,k+2,\ldots
$$
If one assumes additionally (as the numerical evidence suggests, see also Remark \ref{rm:3:4})  that
$$
\hat{p}_0=\ldots=\hat{p}_{k-1}=0,\quad \hat{p}_k>0,
$$
then
\begin{equation}\label{eq3:7}
    \hat{\overline{r}}_\infty=r_k-\mu,\quad \hat p_{k+n}=\frac{\mu^n\hat{p}_k}{\prod_{j=1}^n(r_k-r_{k+j})}\,,\quad n>0,\quad \hat{p}_k=\frac{1}{\displaystyle 1+\sum_{n=1}^\infty\frac{\mu^n}{\prod_{j=1}^n(r_k-r_{k+j})}}\,.
\end{equation}
\end{remark}

\begin{corollary}Assume that the conditions of Lemma \ref{lem:3:10} hold. If, additionally, the limit fitness landscape is such that $r_n=0$ for $n>n_0$ then the generating function $P(\infty)$ is rational and the limit distribution is asymptotically geometric.
\end{corollary}

\begin{corollary}Assume that the conditions of Lemma \ref{lem:3:10} hold. Then formulas \eqref{eq3:7} provide a solution for the non-degenerate limit distribution if and only if
\begin{equation}\label{eq3:8}
    \mu<\mu^\ast=\liminf \sqrt[n]{\prod_{j=1}^n(r_k-r_{k+j})}\leq r_k.
\end{equation}
\end{corollary}
\begin{proof}
Indeed, from \eqref{eq3:7}, the necessary and sufficient condition for the non-degenerate distribution to exist is the convergence of the series
$$
\sum_{n=1}^\infty\frac{\mu^n}{\prod_{j=1}^n(r_k-r_{k+j})}\,,
$$
which, by the Cauchy--Hadamard formula, converges if \eqref{eq3:8} holds and diverges if $\mu>\mu^\ast$.
\end{proof}
\begin{remark}
The critical value $\mu^\ast$ of the mutation rate in \eqref{eq3:8} should be considered the threshold value of the error threshold.
\end{remark}
To illustrate how these general expressions work, consider the fitness landscape
\begin{equation}\label{eq3:8a}
\bs r_\infty=(2,0,1,0,1,0,1,0,\ldots).
\end{equation}
Equation \eqref{eq3:8} predicts that for $\mu<\mu^\ast=\sqrt{2}$ we must have
$$
\hat{\overline{r}}_\infty=2-\mu,
$$
and, for several first coordinates of the equilibrium distribution
$$
\hat{p}_0=\frac{2-\mu^2}{2+\mu}\,,\quad \hat{p}_1=\frac{2-\mu^2}{2+\mu}\cdot\frac{\mu}{2}\,,\quad \hat{p}_2=\frac{2-\mu^2}{2+\mu}\cdot\frac{\mu^2}{2}\,,\quad \hat{p}_3=\frac{2-\mu^2}{2+\mu}\cdot\frac{\mu^3}{4}\,.
$$
Comparison of these exact theoretical predictions with numerical computations are given in Fig.~\ref{fig:7}.
\begin{figure}[!th]
\centering
\includegraphics[width=0.95\textwidth]{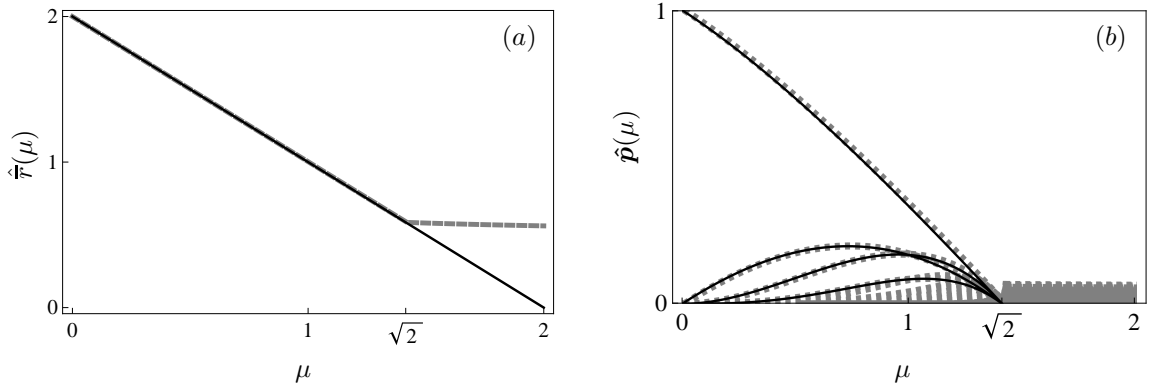}
\caption{Comparison of numerical calculations for the fitness landscape \eqref{eq3:8a} in Example \ref{ex:8} with the theoretical predictions. The black solid lines are the exact solutions for the case $N\to\infty$ and the grey dashed lines are numerical computations for $N=200$. $(a)$ Mean population fitness.  $(b)$ The quasispecies distribution}\label{fig:7}
\end{figure}
\end{example}
\section{Concluding remarks}
We presented an analytical approach to calculate the mutation--selection equilibrium in the Crow--Kimura evolutionary model, which is based on the reformulation of the original eigenvalue problem as a nonlinear functional--differential equation for the unknown probability generating function and on taking a formal limit $N\to\infty$ for the sequence length. This approach provides closed analytical solutions for at least several special fitness landscapes, as we amply illustrated in the previous section. We remark that, to the best of our knowledge, in the existing literature only for two fitness landscapes these equilibrium distributions were written down explicitly, and most attention was concentrated on finding analytical expressions for the mean population fitness (the leading eigenvalue) and for some other population averages. With the advent of sequencing technique, as everyone witnessed for the last two decades, it is now completely feasible to sequence the whole population of quasispecies, and therefore our formulas can be used to further relate theory and experiment in the evolutionary questions. 

While the approach suggested in Section 2 clearly works for all the examples we considered, the conditions are quite difficult to rigorously check and their proof constitutes an independent and deep problem. Our experience tells us that it is quite unlikely to present general (necessary and/or sufficient) conditions, which are easy to check, that would guarantee that our formal limit yields the correct result for an arbitrary fitness landscape. As of now each special case has to be tackled on its own, as we demonstrated for the single peaked landscape (Example~\ref{ex:1}) in \cite{bratus2013linear}. Probably a first realistic step in the search of the general conditions is to prove that the formulas in Example \ref{ex:8} follow rigorously from the assumption on the unique fixed maximum of the fitness landscape. We conclude our text with this open problem.

\appendix
\section{Parametric solutions to the basic eigenvalue problem}\label{app}
The suggested general approach of the generating functions is heuristic. Rigorous proofs can be obtained using the parametric solution method introduced in \cite{bratus2013linear}. In this appendix we give a concise form for the method used in \cite{bratus2013linear} (see Appendix \ref{app:1}) and also introduce a new parametric solution in Appendix \ref{app:2}, which is used to prove the result from Example \ref{ex:7} on the limit form of $\hat{\overline{r}}_\infty$.

Recall that we are interested in finding the dominant eigenvalue and the corresponding positive eigenvector of the problem
$$
(\bs M+\mu \bs Q)\bs{\hat{p}}=\hat{\overline{m}}\, \bs{\hat{p}},
$$
where $\mu\geq 0$, $\bs M=\diag(m_0,\ldots,m_N)$, and the matrix $\bs Q$ has the form
$$
\bs Q=\begin{bmatrix}
               -N & 1 & 0 & 0 & \ldots & \ldots & 0 \\
               N & -N & 2 & 0 & \ldots & \ldots & 0 \\
               0 & N-1 & -N & 3 & \ldots & \ldots & 0 \\
               0 & 0 & N-2 & -N & \ldots & \ldots & 0 \\
               \ldots & \ldots & \ldots & \ldots & \ldots & \ldots & 0 \\
               0 & 0 & \ldots & \ldots& 2 & -N & N \\
               0 & 0 & \ldots & \ldots & 0& 1 & -N \\
             \end{bmatrix}.
$$
We introduce the notations
$$
\bs S=\frac 1N\,\bs Q,\quad \bs R=\frac 1N\,\bs M,\quad \hat{\overline{r}}=\frac 1N\,\hat{\overline{m}}\,.
$$
Then the original eigenvalue problem takes the form
$$
(\bs R+\mu \bs S)\bs{\hat{p}}=\hat{\overline{r}}\, \bs{\hat{p}},
$$
or, after introducing a new parameter $u=\mu/\hat{\overline{r}}$,
\begin{equation}\label{eqA:2}
    \frac{1}{\hat{\overline{r}}}\,\bs R\bs{\hat{p}}+u\bs S\bs{\hat{p}}=\bs{\hat{p}}.
\end{equation}
\subsection{Approach A}\label{app:1} In \cite{bratus2013linear} it was shown that
$$
\bs C^{-1}\bs Q\bs C=-2\diag(0,1,2,\ldots,N),
$$
where $\bs C=(c_{ka})$ is the matrix composed (by columns) of the coefficients of the generating polynomials
$$
P_a(s)=\sum_{k=0}^N c_{ka}s^k=(1-s)^a(1+s)^{N-a},
$$
and possesses the property $\bs C^2=2^N\bs I$, where $\bs I$ is the identity matrix. Using this information, we have
$$
\bs C^{-1}\bs S\bs C=-2\diag\left(0,\frac 1N,\frac 2N,\ldots,1\right).
$$
Equation \eqref{eqA:2} can be written as
$$
\hat{\overline{r}}\,\bs{\hat{p}}=(\bs I-u\bs S)^{-1}\bs R\bs{\hat{p}}.
$$
Using
$$
(\bs I-u\bs S)^{-1}=\bs C^{-1}(\bs I-u\bs C^{-1}\bs S\bs C)^{-1}\bs C=\bs C^{-1}\diag\left(1,\frac{1}{1+\frac{2u}{N}},\frac{1}{1+\frac{4u}{N}},\ldots,\frac{1}{1+\frac{2u N}{N}}\right)\bs C,
$$
we find
$$
(\bs I-u\bs S)^{-1}=\bs F(u)=\bigl(F_{ab}(u)\bigr),\quad F_{ab}=\frac{1}{2^N}\sum_{k=0}^N\frac{c_{ak}c_{kb}}{1+\frac{2ku}{N}}\,.
$$
Therefore, $\hat{\overline{r}}$ is the dominant eigenvalue of the matrix $\bs F(u)\bs R$, and $\bs{\hat{p}}$ is the corresponding eigenvector, both of which can be represented in the parametric form, depending on parameter $u$. The details how to write down the explicit formulas depending on the number of nonzero elements of the vector $\bs m=(m_0,\ldots,m_N)$ are given in \cite{bratus2013linear}.

\subsection{Approach B}\label{app:2}
The equality
$$
\hat{\overline{r}}\,\bs{\hat{p}}=(\bs I-u\bs S)^{-1}\bs R\bs{\hat{p}}
$$
can be written as
\begin{equation}\label{eqA:3}
\hat{\overline{r}}\,\bs{\hat{p}}=\bigl((1+u)\bs I-u(\bs S+\bs I)\bigr)^{-1}\bs R \,\bs{\hat{p}}.
\end{equation}
Let
$$
\bs B=\bs S+\bs I=\begin{bmatrix}
               0 & 1/N & 0 & 0 & \ldots & \ldots & 0 \\
               1 & 0 & 2/N & 0 & \ldots & \ldots & 0 \\
               0 & 1-1/N & 0 & 3/N & \ldots & \ldots & 0 \\
               0 & 0 & 1-2/N & 0 & \ldots & \ldots & 0 \\
               \ldots & \ldots & \ldots & \ldots & \ldots & \ldots & \ldots \\
               0 & 0 & \ldots & \ldots& 2/N & 0 & 1 \\
               0 & 0 & \ldots & \ldots & 0& 1/N & 0 \\
             \end{bmatrix},
$$
i.e., the matrix $\bs B$ is two diagonal, stochastic, and such that
$$
\bs C^{-1}\bs B\bs C=\diag \left(1,1-\frac 2N,1-\frac 4N,\ldots,-1\right).
$$
Direct observations yield that all natural powers of $\bs B$ also will be stochastic. For the even powers $\bs B^{2k}=\bigl(b_{ij}^{(2k)}\bigr)$ we have $b_{ij}^{(2k)}=0$ if $i+j$ is odd, and for the odd powers  $\bs B^{2k-1}=\bigl(b_{ij}^{(2k-1)}\bigr)$ we have $b_{ij}^{(2k-1)}=0$ if $i+j$ is even. Moreover, the properties of $\bs B$ and $\bs C$ imply
\begin{equation*}
\begin{split}
\lim_{k\to\infty}\bs B^{2k}&=\bs C\diag (1,0,\ldots,0,1)\bs C^{-1}=\frac{1}{2^N}\left((1+(-1)^{i+j})\binom{N}{i}\right),\\
\lim_{k\to\infty}\bs B^{2k-1}&=\bs C\diag (1,0,\ldots,0,-1)\bs C^{-1}=\frac{1}{2^N}\left((1-(-1)^{i+j})\binom{N}{i}\right).
\end{split}
\end{equation*}
Using the introduced notations equality \eqref{eqA:3} can be written as
\begin{equation}\label{eqA:4}
    \hat{\overline{r}}\,\bs{\hat{p}}=(1+u)^{-1}\left(\bs I-\frac{u}{1+u}\bs B\right)^{-1}\bs R\,\bs{\hat{p}}=\sum_{k=0}^\infty\frac{u^k}{(1+u)^{k+1}}\bs B^k\bs R\,\bs{\hat{p}},
\end{equation}
where the application of the geometric series can be justified by noting that the 1-norm of the matrix $\frac{u}{1+u}\bs B$ is less than 1 for any $u>0$.

To illustrate how the parametric solution \eqref{eqA:4} works consider the case of the fitness landscape in Example \ref{ex:7}. Let $N=2A$ be even,
$$
\bs r=(0,\ldots,0,1,0,\ldots,0),
$$
where one is at the $A$-th position. In this case
$$
\hat{\overline{r}}=\hat{p}_A,
$$
and \eqref{eqA:4} takes the form
$$
 \hat{\overline{r}}\,\hat p_A=\sum_{k=0}^{\infty}\frac{u^k}{(1+u)^{k+1}}b^{(k)}_{A,A}\,\hat p_A.
$$
Since $b^{(k)}_{A,A}=0$ for odd $k$ and $\hat{\overline{r}}=\hat p_A>0$ then
$$
 \hat{\overline{r}}=\sum_{k=0}^{\infty}\frac{u^{2k}}{(1+u)^{2k+1}}b^{(2k)}_{A,A}.
$$
From the properties of $\bs B$ and $\bs C$ we have
$$
b^{(2k)}_{A,A}=\frac{1}{2^{2A}}\sum_{i=0}^{2A}c_{Ai}c_{iA}\left(1-\frac iA\right)^{2k}.
$$
Using the property that $c_{ij}\binom{N}{j}=c_{ji}\binom{N}{i}$ (see \cite{bratus2013linear} for an easy proof) and the fact that $c_{iA}$ are the coefficients of the polynomial
$$
P_A(s)=\sum_{l=0}^A(-1)^l\binom{A}{l}s^{2l},
$$
we find
$$
c_{Ai}c_{iA}=\begin{cases}0,&i=2l+1,\\
\binom{2l}{l}\binom{2(A-l)}{A-l},&i=2l,
\end{cases}
$$
and therefore
$$
b^{(2k)}_{A,A}=\frac{1}{2^{2A}}\sum_{l=0}^{A}\binom{2l}{l}\binom{2(A-l)}{A-l}\left(1-\frac{2l}{A}\right)^{2k}\,.
$$
Using the approximation
$$
\frac{1}{2^{2n}}\binom{2n}{n}\approx \frac{1}{\sqrt{\pi n}}\,,
$$
we obtain
$$
b^{(2k)}_{A,A}\approx \frac{2}{2^{2A}}\binom{2A}{A}+\frac{1}{\pi}\sum_{l=1}^{A-1}\frac{\left(1-\frac{2l}{A}\right)^{2k}}{\sqrt{l(A-l)}}\,,
$$
or, after taking the limit $A\to\infty$,
$$
\lim_{A\to\infty}b^{(2k)}_{A,A}=\frac{1}{\pi}\int_{0}^1\frac{(1-2x)^{2k}\D x}{\sqrt{x(1-x)}}=\frac{2}{\pi}\int_0^{\pi/2}\cos^{2k}z\D z=\frac{1}{2^{2k}}\binom{2k}{k}.
$$
Therefore, in the limit of the infinite sequence length
$$
\hat{\overline{r}}_\infty=\sum_{k=0}^{\infty}\frac{u^{2k}}{(1+u)^{2k+1}}\frac{1}{2^{2k}}\binom{2k}{k}=\frac{1}{1+u}\frac{1}{\sqrt{1-\frac{u^2}{(1+u)^2}}}=\frac{1}{\sqrt{2u+1}}\,,
$$
where we used the fact that
$$
\frac{1}{\sqrt{1-x^2}}=\sum_{k=0}^\infty \frac{1}{2^{2k}}\binom{2k}{k}x^{2k},\quad |x|<1.
$$
Finally, remembering that $u=\frac{\mu}{\hat{\overline{r}}_\infty}$ gives
$$
\hat{\overline{r}}_\infty=\frac{1}{\sqrt{\frac{2\mu}{\hat{\overline{r}}_\infty}+1}}\,,
$$
from where
$$
\hat{\overline{r}}_\infty=\sqrt{1+\mu^2}-\mu,
$$
as it was stated in Example \ref{ex:7} and proved in \cite{bratus2013linear} using, essentially, approach from Appendix~\ref{app:1}.

\paragraph{Acknowledgements:} 
ASN's research is supported in part by ND EPSCoR and NSF grant \#EPS-0814442.

\end{document}